\documentclass[submission,copyright,creativecommons]{eptcs}
\providecommand{\event}{EXPRESS/SOS 2016} 
\usepackage{breakurl}             
\usepackage{wrapfig}
\usepackage{amscd}
\usepackage{amsthm}
\usepackage{amsmath}
\usepackage{mathtools}
\usepackage{caption}
\usepackage{subfigure}
\usepackage{latexsym,amssymb}
\usepackage{enumitem}
\usepackage{wrapfig}

\usepackage{multirow}
\usepackage{pifont}
\usepackage{rotating}

\usepackage{url} 

\usepackage{soul}               
\soulregister\cite7
\soulregister\ref7
\soulregister\eq7
\soulregister\pageref7
\soulregister{\em}{0}
\soulregister{\bf}{0}
\soulregister{\fig}{7}
\soulregister{\ie}{0}
\soulregister{\etc}{0}
\soulregister{\eg}{0}
\soulregister{\wrt}{0}
\soulregister{\resp}{0}
\soulregister{\cf}{0}

\usepackage{todonotes}          

\usepackage{listings}
\usepackage{hyperref,url}
\hypersetup{pdfborder={0 0 0},
	colorlinks=true,
	linkcolor=black,
	citecolor=black,
	urlcolor=black}
\makeatletter
\def\cleardoublepage{\clearpage\if@twoside \ifodd\c@page\else
    \hbox{}
    \thispagestyle{empty}
    \newpage
    \if@twocolumn\hbox{}\newpage\fi\fi\fi}
\makeatother \clearpage{\pagestyle{plain}\cleardoublepage}

\usepackage{graphicx,color}
\usepackage{tikz}
\usetikzlibrary{matrix,arrows,decorations.pathmorphing,shapes}
\usepackage[explicit]{titlesec}
\newcommand*\chapterlabel{}
\titleformat{\chapter}[display]  
	{\normalfont\bfseries\Huge} 
	{\gdef\chapterlabel{\thechapter\ }}     
 	{0pt} 
 	  {\begin{tikzpicture}[remember picture,overlay]
    \node[yshift=-8cm] at (current page.north west)
      {\begin{tikzpicture}[remember picture, overlay]
        \draw[fill=black] (0,0) rectangle(35.5mm,15mm);
        \node[anchor=north east,yshift=-7.2cm,xshift=34mm,minimum height=30mm,inner sep=0mm] at (current page.north west)
        {\parbox[top][30mm][t]{15mm}{\raggedleft $\phantom{\textrm{l}}$\color{white}\chapterlabel}};  
        \node[anchor=north west,yshift=-7.2cm,xshift=37mm,text width=\textwidth,minimum height=30mm,inner sep=0mm] at (current page.north west)
              {\parbox[top][30mm][t]{\textwidth}{\color{black}#1}};
       \end{tikzpicture}
      };
   \end{tikzpicture}
   \gdef\chapterlabel{}
  } 

\titlespacing*{\chapter}{0pt}{50pt}{30pt}
\titlespacing*{\section}{0pt}{13.2pt}{*0}  
\titlespacing*{\subsection}{0pt}{13.2pt}{*0}
\titlespacing*{\subsubsection}{0pt}{13.2pt}{*0}

\newcounter{myparts}
\newcommand*\partlabel{}
\titleformat{\part}[display]  
	{\normalfont\bfseries\Huge} 
	{\gdef\partlabel{\thepart\ }}     
 	{0pt} 
 	  {\setlength{\unitlength}{20mm}
	  \addtocounter{myparts}{1}
	  \begin{tikzpicture}[remember picture,overlay]
    \node[anchor=north west,xshift=-65mm,yshift=-6.9cm-\value{myparts}*20mm] at (current page.north east) 
      {\begin{tikzpicture}[remember picture, overlay]
        \draw[fill=black] (0,0) rectangle(62mm,20mm);   
        \node[anchor=north west,yshift=-6.1cm-\value{myparts}*20mm,xshift=-60.5mm,minimum height=30mm,inner sep=0mm] at (current page.north east)
        {\parbox[top][30mm][t]{55mm}{\raggedright \color{white}Part \partlabel $\phantom{\textrm{l}}$}};  
        \node[anchor=north east,yshift=-6.1cm-\value{myparts}*20mm,xshift=-63.5mm,text width=\textwidth,minimum height=30mm,inner sep=0mm] at (current page.north east)
              {\parbox[top][30mm][t]{\textwidth}{\raggedleft \color{black}#1}};
       \end{tikzpicture}
      };
   \end{tikzpicture}
   \gdef\partlabel{}
  } 

\newcounter{tempctr}
\newcommand{\breakenumistart}{%
  \setcounter{tempctr}{\value{enumi}}%
  \end{enumerate}%
}
\newcommand{\breakenumiend}{%
  \begin{enumerate}%
  \setcounter{enumi}{\value{tempctr}}%
}

\newtheorem{theorem}{Theorem}[section]

\newtheorem{proposition}[theorem]{Proposition}
\newtheorem{lemma}[theorem]{Lemma}

\theoremstyle{definition}
\newtheorem{definition}[theorem]{Definition}

\newtheorem{example}[theorem]{Example}
\newtheorem{note}[theorem]{Note}

\newcommand{\fig}[2][]{Figure~\ref{fig:#2}\ensuremath{#1}}

\newcommand{\eq}[1]{(\ref{eq:#1})}

\newcommand{\ex}[1]{Example~\ref{ex:#1}}
\newcommand{\secn}[1]{Section~\ref{secn:#1}}

\newcommand{\lem}[1]{Lemma~\ref{lem:#1}}

\newcommand{\prop}[1]{Proposition~\ref{prop:#1}}

\newcommand{\mdash}{---}

\newcommand{\cf}[1][\ ]{cf.#1}
\newcommand{\ie}[1][\ ]{i.e.#1}
\newcommand{\etc}[1][\ ]{etc.#1}
\newcommand{\eg}[1][\ ]{e.g.#1}
\newcommand{\wrt}[1][\ ]{w.r.t.#1}

\newcommand{\bydef}[1]{\ensuremath{\stackrel{def}{#1}}}

\newcommand{\setdef}[2]{\ensuremath{\{{#1}\,|\,{#2}\}}}
\newcommand{\bsetdef}[2]{\ensuremath{\bigl\{\bigl.{#1}\,\bigr|\,{#2}\bigr\}}}
\newcommand{\Setdef}[2]{\ensuremath{\Big\{{#1}\,\Big|\,{#2}\Big\}}}

\newcommand{\sgoesto}[2][]{\ensuremath{\xrightarrow{#2}_{#1}}}
\newcommand{\goesto}[1][]{\sgoesto[]{#1}}

\newcommand{\ngoesto}[1][]{\not\sgoesto{\,\ #1}}

\newcommand{\glue}[1][]{\ensuremath{{\cal G}^{#1}}}

\newcommand{\less}{\prec}

\newcommand{\cA}{\ensuremath{\mathcal{A}}}

\newcommand{\cB}{\ensuremath{\mathcal{B}}}

\newcommand{\cO}{\ensuremath{\mathcal{O}}}

\newcommand{\cR}{\ensuremath{\mathcal{R}}}

\newcommand{\derrule}[3][1]{%
  \ensuremath{%
    \begin{array}{*{#1}{@{\hspace{2mm}}c@{\hspace{2mm}}}}
      #2\\
      \hline
      \multicolumn{#1}{c}{#3}
    \end{array}%
  }%
}

\newlength{\ruleht}

{\end{tikzpicture}}

\title{A Note on the Expressiveness of BIP}
\author{Eduard Baranov and Simon Bliudze
\institute{\'Ecole polytechnique f\'ed\'erale de Lausanne, Station 14, 1015 Lausanne, Switzerland}
\email{firstname.lastname@epfl.ch}
}

\def\titlerunning{A Note on the Expressiveness of BIP}
\def\authorrunning{E. Baranov \& S. Bliudze}

\graphicspath{{graphics/}}

\begin{document}
\maketitle

\begin{abstract}
  We extend our previous algebraic formalisation of the notion of
  component-based framework in order to formally define two
  forms\mdash strong and weak\mdash of the notion of full
  expressiveness.  Our earlier result shows that the BIP
  (Behaviour-Interaction-Priority) framework does not possess the
  strong full expressiveness.  In this paper, we show that BIP has the
  weak form of this notion and provide results detailing weak and strong full expressiveness for classical BIP and several modifications, obtained by relaxing the constraints imposed on priority models.  
\end{abstract}


\section{Introduction}
\label{secn:intro}

In our previous work \cite{BarBliu15-offer}, we have formalised some
of the properties that are desirable for component-based design
frameworks, namely: incrementality, flattening, compositionality and
modularity~\cite{gossler05composition,framework05}.  We have also
discussed the {\em full expressiveness} property, although without
providing a formal definition for it.  The formalisation is based on a
very simple, abstract algebraic definition of the notion of
component-based framework, which we extend below in order to also
provide such formal definition of full expressiveness.

Intuitively, \emph{flattening} requires that, for any component obtained by
hierarchically applying two composition operators to a finite set of
sub-components, there must exist an equivalent component obtained by
applying a single composition operator to the same sub-components.  {\em Full
  expressiveness
  \wrt a given set of operators}\mdash \eg those defined by a
particular Structural Operational Semantics (SOS) rule format\mdash
requires that all operators in that set be expressible as composition
operators in the component-based framework.

In \cite{BarBliu15-offer}, we have studied the satisfaction of the
above properties by BIP (Behaviour-Interaction-Priority), which is a
component-based framework for the design of correct-by-construction
concurrent software and systems based on the separation of concerns
between coordination and computation \cite{Basu:2011,bip06}.

BIP systems consist of components modelled as Labelled Transition
Systems (LTS).  Transitions are labelled by ports, which are used for
synchronisation with other components.  Composition operators defining
such synchronisations are obtained by combining  {\em interaction} and 
{\em priority models}.  Operational semantics of the BIP composition
operators is defined by SOS rules
in a format, which is a restriction of GSOS
\cite{bloom-phd}.
Below we refer to this format as {\em BIP-like SOS}.
We focus on the flattening and full expressiveness \wrt BIP-like SOS of BIP with the
classical semantics defined in \cite{BliSif07-acp-emsoft} and used in the language and code-generation tool-set developed by VERIMAG.\footnote{%
  \url{http://www-verimag.imag.fr/New-BIP-tools.html}
}

In \cite{BarBliu15-offer}, we have provided a counter-example showing
that the classical semantics of BIP does not possess flattening, which
implies that it does not possess full expressiveness \wrt BIP-like SOS either.  
This shows that the often encountered informal statement:
``BIP possesses the expressiveness of the universal glue'' (or its
equivalent in slightly different formulations) is based on an
erroneous proposition in previous work \cite[Proposition
  4]{BliSif08-express-concur}.  
The fundamental reasons for this absence of full
expressiveness lie in the definition of the priority models.
A priority model is a strict partial order on the underlying
interaction model (set of allowed interactions).  This
definition guarantees that applying a priority model does
not introduce deadlocks in the otherwise deadlock-free
system.  However, such deadlocks can be introduced by
certain operators respecting BIP-like SOS.

In \cite{BarBliu15-offer}, we have shown that relaxing the
restrictions on priority models to allow arbitrary relations
on interactions\mdash rather than strict partial orders on
interactions in the interaction model\mdash provides full
expressiveness \wrt the full class of BIP-like SOS
operators.

In this paper, we refine this discussion as follows
\begin{itemize}
\item We formally define two notions\mdash weak and
  strong\mdash of full expressiveness.  Weak full
  expressiveness means that any operator that can be defined
  by a set of BIP-like SOS rules can be expressed as a {\em
    hierarchical composition} of BIP composition operators,
  as opposed to only one composition operator for strong
  full expressiveness.
\item We provide a syntactic characterisation of a subset of
  operators defined by BIP-like SOS rules.
\item We show that BIP has weak full expressiveness with
  respect to this subset of operators.
\item We show that relaxing the partial order restriction in
  the definition of priority models allows us to recover
  weak full expressiveness \wrt the full class of BIP-like
  SOS operators.
\end{itemize}

The rest of the paper is structured as follows: \secn{theory} presents
the algebraic formalisation of the notion ``component-based
framework'' and defines its basic properties, namely flattening,
strong and weak full expressiveness.  \secn{model} introduces the BIP
component-based framework and its formal semantics.  \secn{express}
presents the main results of the paper as stated above.
\secn{related} briefly discusses some related work.  Finally,
\secn{conclusion} concludes the paper.


\section{Algebraic formalisation of component-based frameworks}
\label{secn:theory}

Each component-based design framework can be viewed as
an algebra $\cA$ of components equipped with a semantic mapping
$\sigma$ and an equivalence relation $\simeq$, satisfying a set of
basic properties, which we list below.  More precisely, $\cA$ is an
algebraic structure generated by a {\em behaviour type} $\cB$ and a
set \glue{} of {\em composition (glue) operators}:  
\[
\cA ::= B \ |\ f(C_1,\dots,C_n)\,,
\]
with $B \in \cB$, $C_1, \dots, C_n \in \cA$ and $f \in \glue$.
We call the elements of $\cA$ {\em components} and the elements of
$\cB$ {\em behaviours}.  The algebraic structure $\cA$ represents the
set of all systems constructible within the framework.  Behaviour type
$\cB$ defines the {\em semantic nature} of the components manipulated
by the framework.

The {\em semantic mapping} $\sigma: \cA \rightarrow \cB$ assigns to
each component its meaning in terms of the behaviour type $\cB$.  The
semantic mapping must be consistent: for all $B \in \cB$, must hold
the equality $\sigma(B) = B$.  The semantic mapping is called {\em structural}, if it is
defined by associating to each $n$-ary glue operator $f: \cA^n
\rightarrow \cA$ a corresponding operator $\tilde{f} : \cB^n
\rightarrow \cB$ and putting
\begin{align*}
  &
  \sigma\bigl(f(C_1,\dots,C_n)\bigr) 
  =
  \tilde{f}\bigl(\sigma(C_1),\dots,\sigma(C_n)\bigr)
  \,,
  &
  \text{for all $C_1, \dots, C_n \in \cA$ and $f \in \glue$.}
\end{align*}

Finally, the equivalence relation $\simeq\ \subseteq \cA \times
\cA$\mdash that allows comparing components in terms, for example, of
their functionality, observable behaviour or capability of interaction
with the en\-vi\-ron\-ment\mdash must respect the semantics: for all
$C_1,C_2 \in \cA$, must hold the implication $\sigma(C_1) =
\sigma(C_2) \implies C_1 \simeq C_2$.

In this context, the flattening property mentioned in the introduction
is formalised by requiring that
\begin{multline*}
  \forall i,j \in [1,n]\, (i \leq j),\ \forall C_1, C_2, \dots,
    C_n \in \cA,\ \forall f, g \in \glue,\ \exists h \in \glue:
  \\
  f\bigl(C_1, \dots, C_{i-1}, g(C_i, \dots, C_j), C_{j+1}, \dots, C_n)\bigr)
  \simeq 
  h(C_1, \dots, C_n) 
  \,.
\end{multline*}
In other words, \glue{} must be closed under composition. 
Similarly, given a set of operators $\cO \subseteq
\bigcup_{n=0}^\infty(\cB^n \rightarrow \cB)$, we say that the
component-based framework $(\cA, \sigma, \simeq)$ has \emph{strong
  full expressiveness \wrt $\cO$} iff
\[
  \forall o \in \cO^n,\ \exists \tilde{o} \in \glue:\
  \forall B_1,\dots,B_n \in \cB,\ 
  \sigma(\tilde{o}(B_1,\dots,B_n)) = o(B_1,\dots,B_n)
  \,,
\]
where $\cO^n = \cO \cap (\cB^n \rightarrow \cB)$.  We say that $(\cA,
\sigma, \simeq)$ has \emph{(weak) full expressiveness \wrt $\cO$} iff,
\[
  \forall o \in \cO^n,\ \exists \tilde{o} \in \glue[][Z_1,\dots,Z_n]:\
  \forall B_1,\dots,B_n \in \cB,\ 
  \sigma(\tilde{o}[B_1/Z_1,\dots,B_n/Z_n]) = o(B_1,\dots,B_n)
  \,,
\]
where $\glue[][Z_1,\dots,Z_n]$ is the set of expressions on variables
$Z_1,\dots, Z_n$ obtained by hierarchically applying the glue
operators from $\glue$; whereas $\tilde{o}[B_1/Z_1,\dots,B_n/Z_n] \in
\cA$ is the component obtained by substituting in $\tilde{o}$
the variables $Z_i$ by components $B_i$, for all $i \in [1,n]$.


\section{The BIP component-based framework}
\label{secn:model}

In this section, we briefly recall BIP and its classical operational
semantics, as initially published in \cite{BliSif07-acp-emsoft}.  The
behaviour type in BIP is the set of \emph{Labelled Transition Systems}
(LTS).
  
\begin{definition}
  \label{defn:lts}
  A \emph{labelled transition system} (LTS) is a triple
  $(Q,P,\goesto)$, where $Q$ is a set of \emph{states}, $P$ is a set
  of \emph{ports}, and $\goesto\, \subseteq Q\times (2^P\setminus
  \{\emptyset\}) \times Q$ is a set of \emph{transitions} labelled by
  \emph{interactions}, \ie non-empty sets of ports.
  For $q,q' \in Q$ and $a \in 2^P$, we write $q \goesto[a] q'$ iff
  $(q,a,q') \in\, \goesto$.  A label $a \in 2^P$ is \emph{active} in a
  state $q \in Q$ (denoted $q \goesto[a]$), iff there exists $q' \in
  Q$ such that $q \goesto[a] q'$.  We abbreviate $q
  \ngoesto[a]\ \bydef{=} \lnot (q \goesto[a])$.
\end{definition}

\begin{note}
  \label{rem:indices}
  In the rest of the paper, whenever we speak of a set of LTS
  $B_i=(Q_i, P_i, \sgoesto[i]{})$, for $i\in [1,n]$, we assume that
  all $P_i$ are pairwise disjoint, \ie $i \neq j$ implies $P_i \cap
  P_j = \emptyset$.  We denote $P \bydef{=} \bigcup_{i=1}^n P_i$.
  When the indices are clear from the context, we drop them on
  transition relations and denote write $\goesto$.
\end{note}

Glue operators are separated in two
categories: {\em interaction models} define the sets of allowed {\em
  interactions}, that is synchronisations between the transitions of
their operand components; {\em priority models} define the
scheduling\mdash or conflict resolution\mdash policies, reducing
non-determinism when several synchronisations allowed by the
interaction model are enabled simultaneously.  

\paragraph{Interaction models}
An \emph{interaction model} is a set of interactions $\gamma \subseteq
2^P\setminus \{\emptyset\}$.  The semantics of the application of an
interaction model $\gamma$ is defined by putting
$\sigma(\gamma(B_1,\dots, B_n)) \bydef{=} (Q, P,
\sgoesto[\gamma]{})$, with $Q = \prod_{i=1}^n Q_i$ and the minimal
transition relation $\sgoesto[\gamma]{}$ satisfying the rule
\begin{equation}
  \label{eq:transsem}
  \derrule[3]{
    a \in \gamma
    &
    \Setdef{q_i \goesto[a \cap P_i] q'_i}{i \in I}
    &
    \Setdef{q_i = q'_i}{i \not\in I}
  }{
    q_1\dots q_n \sgoesto[\gamma]{a} q'_1\dots q'_n
  }\,,
\end{equation}
where $I = \setdef{i\in [1,n]}{a \cap P_i \neq \emptyset}$.
Intuitively, this means that an interaction $a$ allowed by
the interaction model $\gamma$ can be fired when all the components
involved in $a$ are ready to fire the corresponding transitions.  All
the components that are not involved in $a$ remain in their current
state.

\paragraph{Priority models}
For a behaviour $B = (Q,P,\goesto)$, a {\em priority model} is a
strict\footnote{%
  As opposed to a (non-strict) partial order, which is a
  reflexive, antisymmetric and transitive relation, a {\em
    strict} partial order is an irreflexive and transitive
  (hence also asymmetric) one.
} partial order $\pi \subseteq 2^P \times 2^P$ (we write $a \less
b$ as a shorthand for $(a,b) \in \pi$).  The semantics of the
application of a priority model $\pi$ is defined by putting
$\sigma(\pi(B)) \bydef{=} (Q, P, \sgoesto[\pi]{})$, with the minimal
transition relation $\sgoesto[\pi]{}$ satisfying the rule
\begin{equation}
  \label{eq:prisem}
  \derrule[2]{
    q \goesto[a] q' &
     \Setdef{q \ngoesto[b]}{a \less b}
  }{
    q \sgoesto[\pi]{a} q'
  }\,.
\end{equation}

Intuitively, this means that an interaction can be fired only if no
higher-priority interaction is enabled. 

Notice that the semantic mapping $\sigma$ defined by \eq{transsem} and
\eq{prisem} is structural, since it is defined by associating to both
interaction and priority models operators on behaviours.

\begin{note}
  \label{rem:interfaces}
  The rules \eq{transsem} and \eq{prisem} defining the semantics of
  BIP operators require that a partition $\bigcup_{i=1}^n P_i = P$ be
  defined, but they do not depend on the specific behaviours $B_1,
  \dots, B_n$.
\end{note}

\begin{definition}
  \label{defn:bip:glue}
  An $n$-ary {\em BIP glue operator} is a triple $\bigl((P_i)_{i=1}^n,
  \gamma, \pi\bigr)$, where $(P_i)_{i=1}^n$ are disjoint sets of ports
  and, denoting $P \bydef{=} \bigcup_{i=1}^n P_i$, the remaining two
  elements $\gamma \subseteq 2^P$ and $\pi \subseteq \gamma \times \gamma$
  are, respectively, interaction and priority models on $P$. 
  (In the remainder of the paper, we omit the sets of ports
  $(P_i)_{i=1}^n$ when they are clear from the context.)
\end{definition}

To simplify the notation, we denote the component obtained by applying
the glue operator $\bigl((P_i)_{i=1}^n,\allowbreak \gamma, \pi\bigr)$ to
sub-components $B_1,\dots,B_n$, by $\pi\gamma(B_1,\dots,B_n)$ instead
of $\bigl((P_i)_{i=1}^n, \gamma, \pi\bigr)(B_1,\dots,B_n)$.
Furthermore, when $\pi = \emptyset$, we write directly
$\gamma(B_1,\dots,B_n)$, omitting $\pi$.

\begin{definition}
  \label{defn:equivalence:classical}
  Two behaviours $B_i=(Q_i,P_i,\goesto)$, for $i=1,2$ are {\em
    equivalent} if $P_1 = P_2$, and the two LTS are bisimilar, \ie
  there exists a bisimulation \cite{park1981caa} relation $R \subseteq
  Q_1 \times Q_2$ total on both $Q_1$ and $Q_2$.
\end{definition}

\begin{figure}
  \centering
  \begin{tabular}[b]{c}
    \subfigure[]{
      \centering
      \scalebox{0.4}{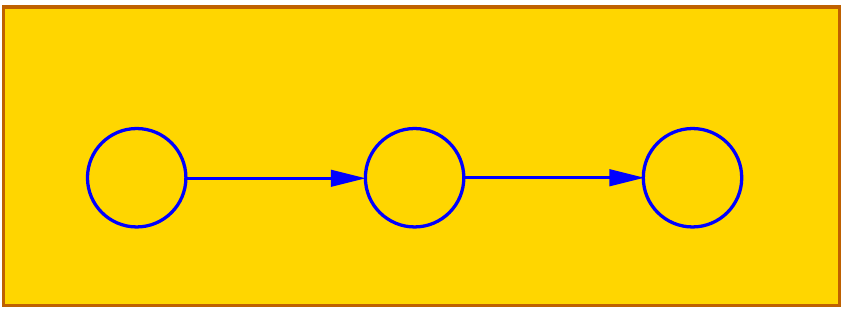}
      \label{fig:priority1}
    }
    \\
    \subfigure[]{
      \centering
      \scalebox{0.4}{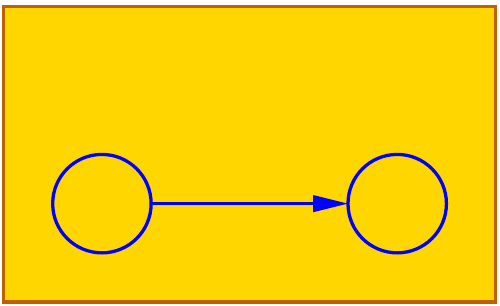}
      \label{fig:priority2}
    }
  \end{tabular}
  \hspace{1cm}
  \subfigure[]{
    \centering
    \scalebox{0.4}{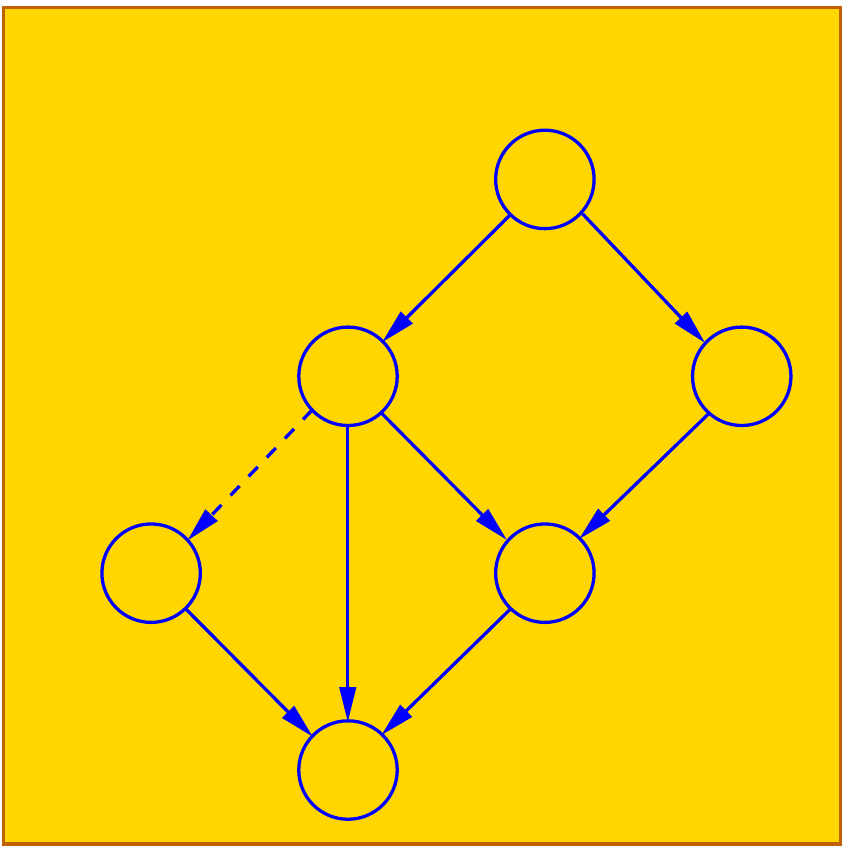}
    \label{fig:priority3}
  }
  \caption{Component behaviours for \ex{priority}}
  \label{fig:priority}
\end{figure}

\begin{example}
  \label{ex:priority}
  Consider the two components $B_1$ and $B_2$ shown in
  Figures~\ref{fig:priority1} and \ref{fig:priority2}, with $P_1 =
  \{p,q\}$ and $P_2 = \{r\}$, and put $\gamma = \{p, q, r, qr\}$ and
  $\pi = \{q\less r\}$.\footnote{%
    To simplify the notation we use the juxtaposition $\gamma = \{p,
    q, r, qr\}$ instead of the set notation $\gamma = \bigl\{\{p\}, \{q\}, \{r\},
    \{q, r\} \bigr\}$ for interactions.  Similarly, we directly write
    $\pi = \{q\less r\}$ instead of $\pi = \{(q, r)\}$
  } The glue operator defined by the combination of the interaction model
  $\gamma$ and the priority model $\pi$ is given by the following four
  rules, obtained by composing
rules of forms \eq{transsem} and \eq{prisem}
and removing premises, whereof satisfaction does not depend on the
state of the operand components (\eg the premise $a \in \gamma$ is
satisfied in all states):
  \begin{equation}
    \label{eq:priority}
    \derrule{
      q_1 \goesto[p] q_1'
    }{
      q_1 q_2 \goesto[p] q_1' q_2
    }
    \,,
    \qquad
    \derrule{
      q_2 \goesto[r] q_2'
    }{
      q_1 q_2 \goesto[r] q_1 q_2'
    }
    \,,
    \qquad
    \derrule[2]{
      q_1 \goesto[q] q_1' &
      q_2 \goesto[r] q_2'
    }{
      q_1 q_2 \goesto[qr] q_1' q_2'
    }
    \,,
    \qquad
    \derrule[2]{
      q_1 \goesto[q] q_1' &
      q_2 \ngoesto[r] 
    }{
      q_1 q_2 \goesto[q] q_1' q_2
    }
    \,.
  \end{equation} 

  The composed component $\pi\gamma(B_1,B_2)$ is shown in
  \fig{priority3}.  The dashed arrow $21 \goesto[q] 31$ shows the
  transition present only in $\gamma(B_1,B_2)$, but not in
  $\pi\gamma(B_1,B_2)$.  Solid arrows show the transitions of
  $\pi\gamma(B_1,B_2)$.

  Among the transitions labelled by $q$, only the transition $22
  \goesto[q] 32$ is enabled and not $21 \goesto[q] 31$
  (\fig{priority3}).  Indeed, the negative premise in the fourth rule
  of \eq{priority}, generated by the priority $q \less r$, suppresses
  the interaction $q$ when a transition labelled $r$ is possible in the
  second component.
  \qed
\end{example}

After merging rules of forms \eq{transsem} and \eq{prisem} and the
simplification by removing the constant premises, all rules used to
define the semantics of BIP glue operators follow the format
\begin{equation}
  \label{eq:format:classical}
  \derrule[3]{
    \Setdef{q_i \goesto[a \cap P_i] q_i'}{i \in I}
    &
    \Setdef{q_i = q_i'}{i \not\in I}
    &
    \Setdef{q_j \ngoesto[b_j^k]}{j \in J, k \in K_j}
  }{
    q_1\dots q_n \goesto[a] q_1'\dots q_n'
  }
  \,,
\end{equation}
where $I = \setdef{i \in [1,n]}{a \cap P_i \neq \emptyset}$, whereas
$J,K_j \subseteq [1,n]$ and, for each $j \in J$ and $k \in K_j$, holds
$b_j^k \in 2^{P_j}$.

Let us now recall an important property of the BIP glue
operators with the above semantics, which was originally
shown in \cite{GosSif04-priority}: application of a priority
model does not introduce deadlocks.

\begin{definition}
  \label{defn:deadlock}
  Let $B = (Q, P, \goesto)$ be a behaviour.  A state $q \in
  Q$ is a {\em deadlock} iff holds $\forall a \subseteq P,\,
  q \ngoesto[a]$.
\end{definition}

\begin{lemma}[\cite{GosSif04-priority}]
  \label{lem:priority}
  Let $B_i = (Q_i, P_i, \goesto)$, for $i\in [1,n]$, be a
  set of behaviours, $\gamma$ and $\pi$ be respectively
  interaction and priority models on $P = \bigcup_{i=1}^n
  P_i$.  A state $q \in \prod_{i=1}^n Q_i$ is a deadlock in
  $\pi\gamma(B_1,\dots,B_n)$ if and only if it is a deadlock
  in $\gamma(B_1,\dots,B_n)$.
\end{lemma}
\begin{proof}
  The ``if'' implication is trivial.  To prove the ``only
  if'' implication, assume that, for some $a \in \gamma$, we
  have $q \sgoesto[\gamma]{a}$.  Let $b \subseteq P$ be an
  interaction, maximal \wrt $\pi$, such that $b \in \gamma$,
  $a \less b$ and $q \sgoesto[\gamma]{b}$.  If such $b$
  exists, holds $q \sgoesto[\pi]{b}$.  Otherwise holds $q
  \sgoesto[\pi]{a}$.  In both cases, $q$ is not a deadlock
  in $\pi\gamma(B_1,\dots,B_n)$.
\end{proof}

Notice that this proof does not rely on $\pi$ being a strict
partial order.  The lemma can be generalised to any {\em
  acyclic} relation $\pi \subseteq \gamma \times \gamma$.


\section{Expressiveness}
\label{secn:express}

We now consider full expressiveness of BIP \wrt the set $\cO$ of
operators defined as pairs $\bigl((P_i)_{i=1}^n, \cR)$, where $n$ is
the arity of the operator, $(P_i)_{i=1}^n$ are pair-wise disjoint sets
of ports and $\cR$ is a set of SOS rules in the format
\eq{format:classical}.
In \cite{BarBliu15-offer}, we have shown that the operator defined by
the following four rules, which respect the format
\eq{format:classical}, cannot be expressed as a BIP glue operator in
the classical semantics:
\begin{equation}
  \label{eq:notbip}
  \derrule[2]{
    q_1 \goesto[p] q_1'
    &
    q_2 \ngoesto[r]
  }{
    q_1 q_2 q_3 \goesto[p] q_1' q_2 q_3
  }
  \,,
  \quad
  \derrule{
    q_1 \goesto[q] q_1'
  }{
    q_1 q_2 q_3 \goesto[q] q_1' q_2 q_3
  }
  \,,
  \quad
  \derrule{
    q_2 \goesto[s] q_2'
  }{
    q_1 q_2 q_3 \goesto[s] q_1 q_2' q_3
  }
  \,,
  \quad
  \derrule[2]{
    q_2 \goesto[r] q_2'
    &
    q_3 \goesto[t] q_3'
  }{
    q_1 q_2 q_3 \goesto[rt] q_1 q_2' q_3'
  }
  \,.
\end{equation}
We conclude that the classical semantics of BIP does not have
neither flattening, nor strong full expressiveness \wrt $\cO$.\footnote{%
  For the details of this example and the associated
  discussion, we refer the reader to \cite{BarBliu15-offer}.
}

Furthermore, the example below shows that the classical semantics of BIP does not have even weak full expressiveness.  

\begin{figure}
  \centering
  \input{graphics/nfebo.pspdftex}
  \caption{Component behaviour for \ex{notbip2}}  
    \label{fig:notbip2}
\end{figure}

\begin{example}
\label{ex:notbip2}
Consider a composition operator defined by the following two rules:
\begin{equation}
  \label{eq:notbip2}
  \derrule[2]{
    q_1 \goesto[p] q_1'
    &
    q_1 \ngoesto[r]
  }{
    q_1 \goesto[p] q_1'
  }
  \,,
  \qquad
  \derrule[2]{
    q_1 \goesto[r] q_1'
    &
    q_1 \ngoesto[p]
  }{
    q_1 \goesto[r] q_1'
  }
  \,,
\end{equation}
applied to the component in \fig{notbip2}. Assume that there exists a hierarchy of BIP glues, such that applying them to the component in \fig{notbip2} results in an equivalent composed component. States 1 and 2 of the composed component have outgoing transitions $p$ and $r$, respectively, thus all interaction models in the glues have to contain both interactions $p$ and $q$. State 3 of the composed component is a deadlock. Interaction models do not forbid any transition from this state and priority models cannot introduce deadlock by \lem{priority}. This contradicts the assumption and, consequently, the set of rules \eq{notbip2} is not expressible in BIP.  
\end{example}

The two fundamental reasons for the lack of
expressiveness are related to the definition of the priority
model:
\begin{itemize}
\item the information used by the priority model refers only
  to interactions authorised by the underlying interaction
  model\mdash all the information about transitions enabled
  in sub-components is lost \cite{BarBliu15-offer};
\item the priority model $\pi$ must be a strict partial
  order.
\end{itemize}

As we explain below, among these two reasons, the first one
is easily addressed to achieve weak, rather than strong,
full expressiveness, whereas the second one presents the
main difficulty.

\paragraph{What can be done without changing the BIP glue?}

Consider an $n$-ary operator $o: \mathit{LTS}^n \rightarrow
\mathit{LTS}$ defined by $(P_i)_{i=1}^n$ and the set of rules
\begin{equation}
 \label{eq:gsos}
 \derrule[3]{
   \Setdef{q_i \goesto[a^l \cap P_i] q_i'}{i \in I^l}
   &
   \Setdef{q_i = q_i'}{i \not\in I^l}
   &
   \Setdef{q_j \ngoesto[b_{j,k}^l]}{j \in J^l, k \in K_j^l}
 }{
   q_1\dots q_n \goesto[a^l] q_1'\dots q_n'
 }
 \,,
 \hspace{1cm}
 \text{for $l \in [1,m]$,}
\end{equation}
where, as above, $I^l = \bsetdef{i \in [1,n]}{a^l \cap P_i \neq
 \emptyset}$.  
For an interaction $a \in \setdef{a^l}{l \in [1,m]}$, denote $R_a
\bydef{=} \setdef{l \in [1,m]}{a = a^l}$ the set of rules with the
conclusion labelled by $a$.  Clearly, for the interaction $a$ to be
inhibited by the negative premises, one such premise must be involved
for each rule in $R_a$.  We denote by $j: R_a \leadsto J$ the {\em
 choice mappings} $j : R_a \rightarrow \bigcup_{l = 1}^m J^l$, such
that $j(l) \in J^l$, for all $l \in R_a$.\footnote{%
  The notion of choice mappings could also be defined as a co-product
  of mappings $\{l\} \rightarrow J^l$ from singleton subsets $\{l\}
  \subseteq R_a$.
}

We define the {\em inhibiting relation} $\pi \subseteq 2^P
\times 2^P$ (where $P = \bigcup_{i=1}^n P_i$) by putting
\begin{equation}
  \label{eq:inhibit}
  \pi\ =\ \bigcup_{l=1}^m\bsetdef{(a^l, b)}{b = \bigcup_{s \in R_{a^l}}b_{j(s),k(s)}^s,\text{ for some } j: R_{a^l} \leadsto J,\ k(s) \in K_{j(s)}^s}\,.
\end{equation}

\begin{figure}
  \centering
  \resizebox{25mm}{30mm}{\input{graphics/cyclebeh.pspdftex}}
  \caption{Component behaviour for \prop{notbipcycle}}  
    \label{fig:cyclebeh}
\end{figure}

\begin{proposition}
\label{prop:notbipcycle}
If $\pi$ has cycles, then the operator $o$ cannot be realised
  by any hierarchical composition of BIP glue operators.
\end{proposition}
\begin{proof}
Consider a cycle in the inhibiting relation $\pi$ : $a_1 \prec a_2 \prec \dots \prec a_l \prec a_1$. 

Let $P = \bigcup_{j=1}^n P_j$, where $P_j = \{p_1^j, \dots, p_m^j\}$.
Let $c_i^j = a_i \cap P_j$ for $i \in [1,l], j \in [1,n]$ and $C_j = \bsetdef{c_i^j}{c_i^j \neq \emptyset}$.
 For each $j$ consider a behaviour as shown in \fig{cyclebeh}. There are no transitions from state $0$; from each state $i$ such that $c_i^j \neq \emptyset$, there is a single transition to state $m_j$ with labels $c_i^j \in C_j$, respectively, and loop transitions in state $m_j$ with labels $c_i^j \in C_j$. 

The composition of such behaviours with the operator $o$ allows a single transition $a_i$ from the state $q_1\dots q_n$, where $q_j = i$ if $c_i^j \neq \emptyset$ or $q_j = 0$ otherwise. In order to allow these transitions, an interaction model of a BIP glue must contain all $a_i$. In the state $q_1\dots q_n$, with $q_j = m_j$, all interactions $a_1, \dots, a_l$ are available. The operator $o$ forbids all of them from this state. Interaction models of BIP glues allow all these interactions and priority models cannot introduce deadlock in this state \lem{priority}. Thus, this system is not expressible in BIP.
\end{proof}

\begin{proposition}
  \label{prop:acyclic}
  If $\pi$ is acyclic, then the operator $o$ can be realised
  by a hierarchical composition of BIP glue operators.
\end{proposition}
\begin{proof}
  Since $\pi$ is acyclic, we can associate a depth $d(a)$ to
  each interaction $a$ involved in $\pi$ as the length of
  the longest path leading to $a$ in the directed acyclic
  graph defined by $\pi$.  Denote $d \bydef{=} \max_a d(a)$.
  Furthermore, for $i \in [1,d]$, denote $\pi_i \bydef{=} \setdef{(a,b)\in
    \pi}{d(a) = i-1}$.

  Clearly all $\pi_i$ are strict partial orders.  
  Furthermore $\pi_i \subseteq \pi \subseteq \gamma_1 \times
  \gamma_1$, for all $i \in [1,d]$ and 
  \begin{align*}
    \gamma_1 &= \gamma_2 \cup 
    \bigcup_{l=1}^m\Setdef{\bigcup_{s \in R_{a^l}}b_{j(s),k(s)}^s}{j: R_{a^l} \leadsto J,\ k(s) \in K_{j(s)}^s}
    \,,
    \\
    \gamma_2 &= \bsetdef{a^l}{l \in [1,m]}
    \,.
  \end{align*}
  Hence, for all $i \in [1,d]$,
  $(\gamma_1, \pi_i)$ is a BIP glue operator.
  
  The operator $o$ is equivalent to the composition $(\gamma_2,\emptyset)\circ (\gamma_1, \pi_d)\circ \dots \circ (\gamma_1, \pi_1)$. We show that for any set of behaviours $B_i = (Q_i, P_i, \goesto{})$, with $i\in [1,n]$, holds $ \sigma\bigl(\gamma_2\bigl(\pi_d\gamma_1\bigl(\dots \pi_1\gamma_1(B_1,\dots, B_n)\bigr)\dots\bigr)\bigr) = o(B_1,\dots, B_n)$. We denote
\begin{align*}
  B_o &= o(B_1,\dots,B_n)\,, 
  &B_{\pi\gamma} &= \sigma\bigl(\gamma_2\bigl(\pi_d\gamma_1\bigl(\dots \pi_1\gamma_1(B_1,\dots, B_n)\bigr)\dots\bigr)\bigr)
  \,.
\end{align*}
The sets of states and ports of these behaviours are the same, thus we only need to check that their transitions coincide.

Let $q_1\dots q_n \goesto[a] q_1'\dots q_n'$ in $B_o$.  This means
that, among the rules defining $o$, \ie for some $l \in [1,m]$, there is a rule
\begin{equation}
  \label{eq:rule1-1}
  \derrule[3]{
    \Setdef{q_i \goesto[a \cap P_i] q_i'}{i \in I^l}
    &
    \Setdef{q_i = q_i'}{i \not\in I^l}
    &
    \Setdef{q_j \ngoesto[b_{j,k}^l]}{j \in J^l, k \in K_j^l}
  }{
    q_1\dots q_n \goesto[a] q_1'\dots q_n'
  }
  \,,
\end{equation}
such that $q_i \goesto[a \cap P_i]$, for all $i \in I$, and $q_j \ngoesto[b_{j,k}^l]$ for all $j \in J^l, k \in K_j^l$.
By construction both $\gamma_1$ and $\gamma_2$ contain $a$. Hence, $a$ is enabled in the state $q_1\dots q_n$ of $\gamma_1(B_1,\dots,B_n)$ and in the same state of $B_{\pi\gamma}$, provided that it is not disabled by any of priorities $\pi_1, \dots, \pi_d$.  Thus, we have to show that no interaction available from this state has higher priority.  By construction, priority rules that contain $a$ in the left-hand side can appear only in $\pi_{d(a)-1}$, thus other priority models cannot block $a$. Priority rules of the form $a \prec b$ have $b = \bigcup_{s \in R_{a}}b_{j(s),k(s)}^s$, for some $j: R_a \leadsto J$ and $k(s) \in K_{j(s)}^s$.  Since all the premises of \eq{rule1-1} are satisfied in $q_1\dots q_n$, interaction $b_{j(l),k(l)}^l$ is disabled.  Hence, $b$ is also disabled.  Thus $q_1\dots q_n \goesto[a] q_1'\dots q_n'$ in $B_{\pi\gamma}$.

Let $q_1\dots q_n \goesto[a] q_1'\dots q_n'$ in $B_{\pi\gamma}$. This
means that both $\gamma_1$ and $\gamma_2$ contain the interaction $a$.
Therefore, by the construction of $\gamma_2$, there is at least one rule 
\begin{equation}
  \label{eq:rule2-1}
  \derrule[3]{
    \Setdef{q_i \goesto[a \cap P_i] q_i'}{i \in I}
    &
    \Setdef{q_i = q_i'}{i \not\in I}
    &
    \Setdef{q_j \ngoesto[b_{j,k}]}{j \in J, k \in K_j}
  }{
    q_1\dots q_n \goesto[a] q_1'\dots q_n'
  }
  \,,
\end{equation}
among the rules defining $o$.  Furthermore, the priority model
$\pi_{d(a)-1}$ contains priorities
of the form $a \prec b$, with $b = \bigcup_{s \in R_{a}}b_{j(s),k(s)}^s$, for all $j: R_a \leadsto J$ and $k(s) \in K_{j(s)}^s$.  Notice that a priority rule $b \prec c$ such that $a \prec b$ cannot appear in priorities $\pi_1, \dots, \pi_{d(a)-1}$ since $d(b) \geq d(a) + 1$.  Assume that none of rules defining $o$, with the conclusion labelled by $a$, applies in $q_1\dots q_n \goesto[a] q_1'\dots q_n'$.  This necessarily means that each of these rules has a negative premise that is not satisfied.  Let $b = \bigcup_{s \in R_{a}}b_{j(s),k(s)}^s$ with $b_{j(s),k(s)}^s$, for all $s \in R_a$, being the labels of dissatisfied premises.  Then $b$ is an enabled interaction in $\gamma_1(B_1,\dots, B_n)$ such that $a \less b$ and $b$ cannot be blocked by priorities $\pi_1, \dots, \pi_{d(a)-1}$. Consequently, $b$ is enabled in $\pi_{d(a)-1}\gamma_1\bigl(\dots \pi_1\gamma_1(B_1,\dots, B_n)\dots\bigr)$ and blocks $a$, which contradicts the assumption $q_1\dots q_n \goesto[a] q_1'\dots q_n'$ in $B_{\pi\gamma}$.  Hence, there is at least one rule of the form \eq{rule2-1} in the definition of $o$ with all premises satisfied in $q_1\dots q_n$ and, therefore, $q_1\dots q_n \goesto[a] q_1'\dots q_n'$ in $B_o$.
\end{proof}

Thus, we conclude that BIP has weak full expressiveness \wrt
the class of BIP-like SOS operators with acyclic inhibiting
relations.

\paragraph{What can be done in the general case?}
In \cite{BarBliu15-offer}, we have proposed the following
notion of relaxed priority model.

\begin{definition}
  \label{defn:relaxed}
  Let $P$ be a set of ports.  A {\em relaxed priority model}
  on $P$ is a relation $\pi \subseteq 2^P \times (2^P
  \setminus \{\emptyset\})$.  A {\em relaxed BIP operator}
  is a triple $\bigl((P_i)_{i=1}^n, \gamma, \pi)$, with $P =
  \bigcup_{i=1}^n P_i$, such that $\gamma \subseteq 2^P \setminus \{\emptyset\}$ is an
  interaction model and $\pi \subseteq \gamma \times \gamma$
  is a relaxed priority model.
\end{definition}

The semantics of relaxed priority models is defined exactly
as that of classical priority models, by \eq{prisem}.
Notice that we do not require the relation $\pi$ to be
acyclic.  If all interactions involved in a cyclic
dependency in $\pi$ are enabled simultaneously, they block
each other, potentially introducing a deadlock.

Given a BIP-like SOS operator $o$, we consider its inhibiting relation $\pi$ (see \eq{inhibit}) and the
interaction models $\gamma_1, \gamma_2$ as in the proof of
\prop{acyclic}. 
Since $\pi \subseteq \gamma_1 \times \gamma_1$, the operator
$(\gamma_1, \pi)$ is a relaxed BIP operator.
The operator $o$ is then equivalent to the
composition $(\gamma_2, \emptyset) \circ (\gamma_1, \pi)$,
where $\pi$ is considered as a relaxed priority model.

\begin{proposition}
  \label{prop:express:weak}
For any set of behaviours $B_i = (Q_i, P_i, \goesto{})$, with $i\in [1,n]$, holds 
\[
\sigma\bigl(\gamma_2\bigl(\pi\gamma_1(B_1,\dots, B_n)\bigr)\bigr) = o(B_1,\dots, B_n)\,.
\]
\end{proposition}
\begin{proof} 
For a set of behaviours $B_i = (Q_i, P_i, \goesto{})$, with $i\in
[1,n]$, denote
\begin{align*}
  B_o &= o(B_1,\dots,B_n)\,, 
  &B_{\pi\gamma} &= \sigma\bigl(\gamma_2\bigl(\pi\gamma_1(B_1,\dots, B_n)\bigr)\bigr)
  \,.
\end{align*}
The sets of states and ports of these behaviours are the same, thus we only need to check that their transitions coincide.

Let $q_1\dots q_n \goesto[a] q_1'\dots q_n'$ in $B_o$.  This means
that, among the rules defining $o$, \ie for some $l \in [1,m]$, there is a rule
\begin{equation}
  \label{eq:rule1}
  \derrule[3]{
    \Setdef{q_i \goesto[a \cap P_i] q_i'}{i \in I^l}
    &
    \Setdef{q_i = q_i'}{i \not\in I^l}
    &
    \Setdef{q_j \ngoesto[b_{j,k}^l]}{j \in J^l, k \in K_j^l}
  }{
    q_1\dots q_n \goesto[a] q_1'\dots q_n'
  }
  \,,
\end{equation}
such that $q_i \goesto[a \cap P_i]$, for all $i \in I$, and $q_j \ngoesto[b_{j,k}^l]$ for all $j \in J^l, k \in K_j^l$.
By construction both $\gamma_1$ and $\gamma_2$ contain $a$. Hence, $a$ is enabled in the state $q_1\dots q_n$ of $\gamma_1(B_1,\dots,B_n)$ and in the same state of $\gamma_2\bigl(\pi\gamma_1(B_1,\dots, B_n)\bigr)$, provided that it is not disabled by the priority $\pi$.  Thus, we have to show that no interaction available from this state has higher priority.  Priority rules in $\pi$ that contain $a$ are of the form $a \prec b$, with $b = \bigcup_{s \in R_{a}}b_{j(s),k(s)}^s$, for some $j: R_a \leadsto J$ and $k(s) \in K_{j(s)}^s$.  Since all the premises of \eq{rule1} are satisfied in $q_1\dots q_n$, interaction $b_{j(l),k(l)}^l$ is disabled.  Hence, $b$ is also disabled.  Thus $q_1\dots q_n \goesto[a] q_1'\dots q_n'$ in $B_{\pi\gamma}$.

Let $q_1\dots q_n \goesto[a] q_1'\dots q_n'$ in $B_{\pi\gamma}$. This
means that bot $\gamma_1$ and $\gamma_2$ contain the interaction $a$.
Therefore, by the construction of $\gamma_2$, there is at least one rule 
\begin{equation}
  \label{eq:rule2}
  \derrule[3]{
    \Setdef{q_i \goesto[a \cap P_i] q_i'}{i \in I}
    &
    \Setdef{q_i = q_i'}{i \not\in I}
    &
    \Setdef{q_j \ngoesto[b_{j,k}]}{j \in J, k \in K_j}
  }{
    q_1\dots q_n \goesto[a] q_1'\dots q_n'
  }
  \,,
\end{equation}
among the rules defining $o$.  Furthermore, the priority model
$\pi$ has to contain priorities
of the form $a \prec b$, with $b = \bigcup_{s \in R_{a}}b_{j(s),k(s)}^s$, for all $j: R_a \leadsto J$ and $k(s) \in K_{j(s)}^s$.  Assuming now that none of rules defining $o$, with the conclusion labelled by $a$, applies in $q_1\dots q_n \goesto[a] q_1'\dots q_n'$.  Since $q_1\dots q_n \goesto[a] q_1'\dots q_n'$ in $B_{\pi\gamma}$, this necessarily means that each of these rules has a negative premise that is not satisfied.  Let $b = \bigcup_{s \in R_{a}}b_{j(s),k(s)}^s$ with $b_{j(s),k(s)}^s$, for all $s \in R_a$, being the labels of dissatisfied premises.  Then $b$ is an enabled interaction such that $a \less b$, which contradicts the assumption $q_1\dots q_n \goesto[a] q_1'\dots q_n'$ in $B_{\pi\gamma}$.  Hence, there is at least one rule of the form \eq{rule2} in the definition of $o$ with all premises satisfied in $q_1\dots q_n$ and, therefore, $q_1\dots q_n \goesto[a] q_1'\dots q_n'$ in $B_o$.
\end{proof}

Thus, we conclude that BIP with relaxed priority models has
weak full expressiveness \wrt the set of all BIP-like SOS
operators.

Notice that the relaxed priority model does not allow
recovering strong full expressiveness.  For instance,
consider the operator defined by the single rule
\begin{equation}
\label{eq:notstrong}
  \derrule[2]{
    q_1 \goesto[p] q_1'
    &
    q_1 \ngoesto[r]
  }{
    q_1 \goesto[p] q_1'
  }\,,
\end{equation} 
applied to the behaviour in \fig{notbip2}. The composed component has a single transition $1 \goesto[p] 3$. The interaction model of BIP cannot contain $r$, as it is not possible to exclude transition $2 \goesto[r] 3$ with a priority model. The transition $3 \goesto[p] 3$ has to be excluded by the priority model, however it cannot use $r$ in the priority relation.

Further relaxation of the definition of the BIP
operator by removing the restriction $\pi \subseteq \gamma
\times \gamma$ requires a slight modification of the semantics. Clearly, the behaviour $\gamma(B_1, \dots, B_n)$ does not have transitions that are not in $\gamma$ and priority rules that can be applied to this behaviour are in $\gamma \times \gamma$. Thus, we need to apply interaction and priority models simultaneously. The semantics of the simultaneous application of an
interaction model $\gamma$ and a priority model $\pi$ is defined by putting
$\sigma(\pi\gamma(B_1,\dots, B_n)) \bydef{=} (Q, P,
\sgoesto[\pi\gamma]{})$, with $Q = \prod_{i=1}^n Q_i$ and the minimal
transition relation $\sgoesto[\pi\gamma]{}$ inductively defined by the set of rules
\begin{equation}
  \label{eq:transsem2}
  \left\{\left.
  \derrule[3]{
    \Setdef{q_i \goesto[a \cap P_i] q'_i}{i \in I}
    &
    \Setdef{q_i = q'_i}{i \not\in I}
    &
    \Setdef{q_j \ngoesto[b \cap P_j]}{b \in K_a}
  }{
    q_1\dots q_n \sgoesto[\pi\gamma]{a} q'_1\dots q'_n
  }
  \ \right|\
    a \in \gamma, j : K_a \leadsto [1,n]
  \right\}
  \,,
\end{equation}
where $I = \setdef{i\in [1,n]}{a \cap P_i \neq \emptyset}$, $K_a = \{b | a \prec b\}$ and $j : K_a \leadsto [1,n]$ is a choice mapping $j : K_a \rightarrow [1,n]$, such that, for all $b \in K_a$, holds $b \cap P_{j(b)} \neq \emptyset$.

With this relaxation we obtain strong full expressiveness, since
the operator $o$ is then clearly equivalent to $(\gamma_2, \pi)$.

\paragraph{What cannot be achieved?}
Consider another relaxation of the definition of BIP glue
operators, by considering operators $\bigl((P_i)_{i=1}^n,
\gamma, \pi\bigr)$, with $P = \bigcup_{i=1}^n P_i$, such
that the priority model $\pi \subseteq 2^P \times (2^P
\setminus \{\emptyset\})$ is a strict partial order, without
requiring that it refer only to interactions (\ie we do not
impose $\pi \subseteq \gamma \times \gamma$).
%
This relaxation does not recover even weak full
expressiveness \wrt BIP-like SOS operators.  Indeed,
\ex{notbip2} is still no expressible.

\section{Related Work}
\label{secn:related}

The results in this paper build mainly on our previous work.  However,
the following related work should also be mentioned.

Usually, comparison between formalisms and models is by flattening
structure and reduction to a behaviorally equivalent model, \eg
automata and Turing machine.  In this manner, all finite state
formalisms turn out to be expressively equivalent independently of the
features used for the composition of behaviors.  Many models and languages
are Turing-expressive, while their coordination capabilities are
tremendously different. \cite{BliSif08-express-concur}

A first framework formally capturing meanings of expressiveness for
sequential programming languages and taking into account not only the
semantics but also the primitives of languages was provided in
\cite{fell90-exp-power-pl}.  It allows formal reasoning about and
distinguishing between {\em core elements} of a language and {\em
  syntactic sugar}.  Although a number of studies have taken a similar
approach in the context of concurrency, we will only point to
\cite{GORLA20101031} and the references therein.  The key difference
of our approach lies in the strong separation between the computation
and coordination aspects of the behaviour of concurrent systems.
Indeed, we consider that all sequential computation resides within the
components of the system that are not subject to any kind of
modification.  Thus, we focus on the following question: {\em what
  system behaviour can be obtained by coordination of a given set of
  concurrent components?}  In particular, this precludes the
expression of parallel composition by choice operators, as in the
expansion law \cite{milner89}.

An extensive overview of SOS formats is provided in \cite{sos20ya},
including some results comparing their expressiveness.  More results
comparing different formats of SOS can be found in \cite{ordsos}.  The
expressiveness property is closely related to the translation between
languages.  One of the definitions of encoding compared with other
approaches can be found in \cite{van2012musings}.  It should be noted,
however, that the above mentioned separation of concerns principle
also leads to a very simple rule format.  Indeed, the format that we
consider is a small subset of GSOS.  Our focus in this paper, is more
on the expressiveness of coordination mechanism provided by BIP than
on that of the various SOS rule features.

There exist several works comparing BIP with various connector frameworks. A comparative study of three
connector frameworks\mdash tile model \cite{montanari06}, wire
calculus \cite{sobocinski09-wire} and BIP \cite{bip06}\mdash was presented in~\cite{bruni12-tiles-wire-BIP}. Recently an attempt to relate BIP and Reo has been done \cite{dokter2015relating}. 
From the operational semantics perspective, these comparisons only take
in account operators with positive premises.  In particular, priority
in BIP is not considered.

Finally, in our formalisation of component-based frameworks, we rely
on the notion of ``behaviour type''.  This can cover a very large
spectrum, ranging from programs and labelled transition systems,
through OSGi bundles and browser plug-ins, to systems of differential
equations \etc[] Behaviour types can be organised in type systems and
studied separately, as, for example, in the co-algebra theory
\cite{Rutten00}.  However, this notion should be distinguished, for
instance, from classes in object-oriented programming or session
\cite{Dezani10-session-types,honda98-session-types} and behavioural
\cite{Huttel16-foundations-types} types for communication protocols.
For instance, the notion of a class could be compared to that of a
behaviour type in our sense as follows: a program would typically comprise a
multitude of classes, whereas a component framework has only one
underlying behaviour type.  Although, in principle, component-based
frameworks can be heterogeneous, \eg Ptolemy~II \cite{Ptolemy}, that
is rely on several distinct behaviour types for the design process,
those aimed at the design of executable systems must have an
underlying unifying behaviour type allowing the study and manipulation of
a system as a whole.


\section{Conclusion}
\label{secn:conclusion}

Our previous investigations \cite{BarBliu15-offer} of several
properties that we consider fundamental for component-based design
frameworks have revealed that the often encountered informal
statement: ``BIP possesses the expressiveness of the universal glue''
(or its equivalent in slightly different formulations) is based on an
erroneous proposition in previous work \cite[Proposition
  4]{BliSif08-express-concur}.  We have, therefore, undertook an
additional study of BIP expressiveness, whereof the results have been
presented in this paper.

To achieve this goal, we rely on the algebraic formalisation of the
notion of component-based design framework introduced in
\cite{BarBliu15-offer}.  We have defined two new properties,
\emph{weak} and \emph{strong full expressiveness} \wrt a given set of
composition operators, which characterise whether these can be
expressed by using the composition operators of the component-based
framework under consideration.  These two properties are very close to
the {\em weakly more expressive} and {\em strongly more expressive}
preorders introduced in \cite{BliSif08-express-concur}.  In
particular, for a component-based framework $(\cA, \sigma, \simeq)$,
with the underlying structure $\cA$ being generated by a set of glue
operators $\glue$, the strong full expressiveness property \wrt a set
of operators $\cO$ coincides with the statement that $\glue$ is strongly
more expressive than $\cO$ in terms of \cite{BliSif08-express-concur}.
However, the formal definition that we have provided in \secn{theory}
is novel and has the advantage of fitting elegantly with that of the
component-based design frameworks in \cite{BarBliu15-offer}.
Furthermore, the notion of weak full expressiveness is different from
the {\em weakly more expressive} preorder: the former relaxes the
strong form of the property by allowing hierarchical composition of
glue operators, whereas the latter considers only flat operators, but
allows a limited use of additional coordinating behaviour.  Studying
the combination of the two relaxations could be an interesting
direction for future work.

We have studied the weak and full expressiveness of BIP \wrt operators
defined by SOS rules in a particular format, which we call {\em
  BIP-like SOS}.  The set of all the operators that can be defined in
this format is the ``universal glue'', \wrt which full expressiveness
has been erroneously claimed in \cite{BliSif08-express-concur}.

We observe that there are two obstacles to achieving strong full
expressiveness: 1)~a priority model is required to be a {\em strict
  partial order} on interactions and 2)~by the definition of the BIP
operational semantics, priorities can only be applied to interactions
that appear in the interaction model.  The combination of these two requirements ensures that priorities cannot introduce new deadlocks.  However, negative premises in BIP-like SOS rules\mdash which correspond to priorities in BIP glue operators\mdash can introduce deadlocks.  To characterise this situation, we consider, for a set of BIP-like SOS rules, a corresponding {\em inhibiting relation}.  In order to introduce deadlocks, this relation must have cycles.  We show that BIP glue operators have weak full expressiveness \wrt BIP-like SOS operators that have acyclic inhibiting relations, with at most $d+1$ layers of glue necessary to encode a BIP-like SOS operator, whereof the depth of the inhibiting relation is $d$.

A relaxation of both of the above
requirements together recovers {\em strong} full expressiveness \wrt all
BIP-like SOS.  However, it calls for a definition of the operational
semantics of BIP glue operators, which combines the interaction and
the priority models, as opposed to the classical definition, where the
interaction model is applied first, then the priority model is applied
to the resulting component.  

A relaxation of only the first requirement, which does not require any
other modifications of the BIP semantics, leads to {\em weak} full
expressiveness \wrt the set of all BIP-like SOS operators.  Moreover,
we have shown that at most two layers of glue are necessary to encode
any operator.


As mentioned above, studying the combination of the two weak forms of
full expressiveness\mdash allowing both hierarchical glue and limited
use of additional coordinating behaviour\mdash could be an interesting
direction for future work.  Another direction for future work would
consist in exploring the expressiveness of the full BIP framework,
including the data manipulation and transfer, which has been recently
formalised in \cite{BBJS14-internalisation}.  Finally, a third
extension could consist in studying larger SOS formats, including, for
instance, {\em witness premises}, \ie positive premises that allow
testing the possibility of an action that does not, however,
contribute to the conclusion of the rule.


\bibliographystyle{eptcs}
\bibliography{bip,connectors,express,glue}

\end{document}